\documentclass[a4paper,USenglish]{lipics-v2018}

\usepackage{microtype}%if unwanted, comment out or use option "draft"

\graphicspath{{./Figures/}}%helpful if your graphic files are in another directory

\usepackage{xspace}
\usepackage{amssymb,amsmath,amsthm}

\title{Lower Bounds for Dynamic Programming on Planar Graphs of Bounded Cutwidth}

\author{Bas A.M. van Geffen}{University of Oxford}{bas.vangeffen@kellogg.ox.ac.uk}{}{}
\author{Bart M.P. Jansen}{Eindhoven University of Technology}{b.m.p.jansen@tue.nl}{http://orcid.org/0000-0001-8204-1268}{Supported by NWO Gravitation grant ``Networks'' and NWO Veni grant ``Frontiers in Parameterized Preprocessing''.}
\author{Arnoud A.W.M. de Kroon}{University of Oxford}{noud.kroon@keble.ox.ac.uk}{}{}
\author{Rolf Morel}{University of Oxford}{rolf.morel@sjc.ox.ac.uk}{}{}

\authorrunning{B.A.M. van Geffen, B.M.P. Jansen, A.A.W.M. de Kroon, and R. Morel}%mandatory. First: Use abbreviated first/middle names. Second (only in severe cases): Use first author plus 'et. al.'

\Copyright{Bas A.M. van Geffen, Bart M.P. Jansen, Arnoud A.W.M. de Kroon, and Rolf Morel}%mandatory, please use full first names. LIPIcs license is "CC-BY";  http://creativecommons.org/licenses/by/3.0/

\subjclass{%
,
}


\EventEditors{John Q. Open and Joan R. Access}
\EventNoEds{2}
\EventLongTitle{42nd Conference on Very Important Topics (CVIT 2016)}
\EventShortTitle{IPEC 2018}
\EventAcronym{IPEC}
\EventYear{2016}
\EventDate{December 24--27, 2016}
\EventLocation{Little Whinging, United Kingdom}
\EventLogo{}
\SeriesVolume{42}
\ArticleNo{23}
\nolinenumbers %uncomment to disable line numbering
\hideLIPIcs  %uncomment to remove references to LIPIcs series (logo, DOI, ...), e.g. when preparing a pre-final version to be uploaded to arXiv or another public repository
\theoremstyle{plain}

\newtheorem{claim}[theorem]{Claim}
\newtheorem{proposition}[theorem]{Proposition}
\newtheorem{observation}[theorem]{Observation}
\newtheorem*{seth}{Strong Exponential Time Hypothesis}

\newcommand{\qSAT}{\textsc{$q$-SAT}\xspace}
\newcommand{\IS}{\ensuremath{\mathrm{\textsc{is}}}\xspace}

\newcommand{\DS}{\ensuremath{\mathrm{\textsc{ds}}}\xspace}
\newcommand{\Hvc}{\ensuremath{H_{\mathrm{\textsc{vc}}}}\xspace}

\newcommand{\Oh}{\ensuremath{\mathcal{O}}\xspace}

\newcommand{\tw}{\ensuremath{\mathrm{tw}}\xspace}
\newcommand{\cutw}{\ensuremath{\mathrm{ctw}}\xspace}
\newcommand{\pw}{\ensuremath{\mathrm{pw}}\xspace}

\newcommand{\optIS}{\ensuremath{\textsc{opt}_{\textsc{\IS}}}\xspace}

\newcommand{\optDS}{\ensuremath{\textsc{opt}_{\textsc{\DS}}}\xspace}

\newcommand{\bigmid}{\,\big\vert\,}
\newcommand{\yes}{\textsc{yes}\xspace}

\let\plainqed\qedsymbol
\newcommand{\claimqed}{$\lrcorner$}

\newenvironment{claimproof}{\begin{proof}\renewcommand{\qedsymbol}{\claimqed}}{\end{proof}\renewcommand{\qedsymbol}{\plainqed}}

\begin{document}

\maketitle

\begin{abstract}
Many combinatorial problems can be solved in time~$\Oh^*(c^{tw})$ on graphs of treewidth~$tw$, for a problem-specific constant~$c$. In several cases, matching upper and lower bounds on~$c$ are known based on the Strong Exponential Time Hypothesis (SETH). In this paper we investigate the complexity of solving problems on graphs of bounded cutwidth, a graph parameter that takes larger values than treewidth. We strengthen earlier treewidth-based lower bounds to show that, assuming SETH, \textsc{Independent Set} cannot be solved in~$O^*((2-\varepsilon)^{\cutw})$ time, and \textsc{Dominating Set} cannot be solved in~$O^*((3-\varepsilon)^{\cutw})$ time. By designing a new crossover gadget, we extend these lower bounds even to planar graphs of bounded cutwidth or treewidth. Hence planarity does not help when solving \textsc{Independent Set} or \textsc{Dominating Set} on graphs of bounded width. This sharply contrasts the fact that in many settings, planarity allows problems to be solved much more efficiently.
\end{abstract}

\section{Introduction}
Dynamic programming on graphs of bounded treewidth is a powerful tool in the algorithm designer's toolbox, which has many applications~(cf.~\cite{BodlaenderK08}) and is captured by several meta-theorems~\cite{Courcelle90,Pilipczuk11}. Through clever use of techniques such as M\"obius transformation, fast subset convolution~\cite{BjorklundHKK07,RooijBR09}, cut \& count~\cite{CyganNPPRW11}, and representative sets~\cite{BodlaenderCKN15,CyganKN13,FominLPS16}, algorithms were developed that can solve numerous combinatorial problems on graphs of treewidth~$tw$ in~$\Oh^*(c^{tw})$ time, for a problem-specific constant~$c$. In recent work~\cite{LokshtanovMS11}, it was shown that under the \emph{Strong Exponential Time Hypothesis} (SETH, see~\cite{ImpagliazzoP01,ImpagliazzoPZ01}), the base of the exponent~$c$ achieved by the best-known algorithm is actually optimal for \textsc{Dominating Set} ($c=3$) and \textsc{Independent Set} ($c=2$), amongst others. This prompts the following questions:
\begin{enumerate}
	\item Do faster algorithms exist for bounded-treewidth graphs that are \emph{planar}? \label{q:planar}
	\item Do faster algorithms exist for a more restrictive graph parameter, such as cutwidth? \label{q:cutwidth}
\end{enumerate}
\noindent It turns out that these questions are related, because the nature of cutwidth allows crossover gadgets to be inserted to planarize a graph without increasing its width significantly. 

Before going into our results, we briefly motivate these questions. There is a rich bidimensionality theory (cf.~\cite{DemaineH08}) of how the planarity of a graph can be exploited to obtain better algorithms than in the nonplanar case, leading to what has been called the \emph{square-root phenomenon}~\cite{Marx13}: in several settings, parameterized algorithms on planar graphs can be faster by a square-root factor in the exponent, compared to their nonplanar counterparts. Hence it may be tempting to believe that problems on bounded-width graphs can be solved more efficiently when they are planar. Lokshtanov et al.~\cite[\S 9]{LokshtanovMS11} explicitly ask whether their SETH-based lower bounds continue to apply for planar graphs. The same problem is posed by Baste and Sau~\cite[p.~3]{BasteS15} in their investigation on the influence of planarity when solving connectivity problems parameterized by treewidth. This motivates question~\ref{q:planar}.

When faced with lower bounds for the parameterization by treewidth, it is natural to investigate whether these continue to hold for more restrictive graph parameters. We work with the parameter cutwidth since it is one of the classic graph layout parameters (cf.~\cite{DiazPS02}) which takes larger values than treewidth~\cite{KorachS93}, and has been the subject of frequent study~\cite{GiannopoulouPRT16,ThilikosSB05,ThilikosSB05a}. In their original work, Lokshtanov et al.~\cite{LokshtanovMS11} showed that their lower bounds also hold for pathwidth instead of treewidth. However, the parameterization by \emph{cutwidth} has so far not been considered, which leads us to question~\ref{q:cutwidth}. (See Section~\ref{sec:preliminaries} for the definition of cutwidth.)

\subparagraph{Our results} We answer questions~\ref{q:planar} and~\ref{q:cutwidth} for the problems \textsc{Independent Set} and \textsc{Dominating Set}, which are formally defined in Section~\ref{sec:preliminaries}. Our conceptual contribution towards answering question~\ref{q:planar} comes from the following insight: any graph~$G$ can be drawn in the plane (generally with crossings) such that the graph~$G'$ obtained by replacing each crossing by a vertex of degree four, does not have larger cutwidth than~$G$. Hence the property of having bounded cutwidth can be preserved while planarizing the graph, which was independently\footnote{We learned of Eppstein's result while a previous version of this work was under submission at a different venue; see Footnote 2 in~\cite{Eppstein17a}. Our previous manuscript, cited by Eppstein, was later split into two separate parts due to its excessive length. The present paper is one part, and~\cite{JansenN17} is the other.} discovered by Eppstein~\cite{Eppstein17a}. When we planarize by replacing each crossing by a planar crossover gadget~$H$ instead of a single vertex, then we obtain~$\cutw(G') \leq \cutw(G) + \cutw(H) + 4$ if the endpoints of the crossing edges each obtain at most one neighbor in the crossover gadget. This gives a means to reduce a problem instance on a general graph of bounded cutwidth to a planar graph of bounded cutwidth, if a suitable crossover gadget is available. The parameter cutwidth is special in this regard: one cannot planarize a drawing of~$K_{3,n}$ while keeping the pathwidth or treewidth constant~\cite{Eppstein16,Eppstein17a}.

For the \textsc{Independent Set} problem, the crossover gadget developed by Garey, Johnson, and Stockmeyer~\cite{GareyJS76} can be used in the process described above. Together with the observation that the SETH-based lower bound construction by Lokshtanov et al.~\cite{LokshtanovMS11} for the treewidth parameterization also works for the cutwidth parameterization, this yields our first result.\footnote{The analogous lower bound of~$\Omega((2-\varepsilon)^{k}$) for solving \textsc{Independent Set} on planar graphs of pathwidth~$k$ was already observed by Jansen and Wulms~\cite{JansenW16}, based on an elaborate ad-hoc argument.} 

\newcommand{\indsetthm}{Assuming SETH, there is no~$\varepsilon > 0$ such that \textsc{Independent Set} on a planar graph~$G$ given along with a linear layout of cutwidth~$k$ can be solved in time~$\Oh^*((2-\varepsilon)^k)$.}
\begin{theorem} \label{thm:independentset:lb}
\indsetthm
\end{theorem}

For the \textsc{Dominating Set} problem, more work is needed to obtain a lower bound for planar graphs of bounded cutwidth. While the lower bound construction of Lokshtanov et al.~\cite{LokshtanovMS11} also works for the parameter cutwidth after a minor tweak, no crossover gadget for the \textsc{Dominating Set} problem was known. Our main technical contribution therefore consists of the design of a crossover gadget for \textsc{Dominating Set}, which we believe to be of independent interest. Together with the framework above, this gives our second result.

\newcommand{\domsetthm}{Assuming SETH, there is no~$\varepsilon > 0$ such that \textsc{Dominating Set} on a planar graph~$G$ given along with a linear layout of cutwidth~$k$ can be solved in time~$\Oh^*((3-\varepsilon)^k)$.}
\begin{theorem} \label{thm:domset:planar:lb}
\domsetthm
\end{theorem}

Since any linear ordering of cutwidth~$k$ can be transformed into a tree decomposition of width at most~$k$ in polynomial time (cf.~\cite[Theorem 47]{Bodlaender98}), the lower bounds of Theorems~\ref{thm:independentset:lb} and~\ref{thm:domset:planar:lb} also apply to the parameterization by treewidth. Hence our work resolves the question raised by Lokshtanov et al.~\cite{LokshtanovMS11} and by Baste and Sau~\cite{BasteS15} whether the SETH-lower bounds for \textsc{Independent Set} and \textsc{Dominating Set} parameterized by treewidth also apply for planar graphs.

\subparagraph*{Organization}
In Section~\ref{sec:definitions} we provide preliminaries.
In Section~\ref{sect:planarization} we present a general theorem for planarizing graphs of bounded cutwidth, using a crossover gadget. It leads to a proof of Theorem~\ref{thm:independentset:lb}. In Section~\ref{sec:lowds} we prove Theorem~\ref{thm:domset:planar:lb}.
Finally, we provide some conclusions in Section~\ref{sec:conc}. Due to space restrictions, proofs for statements marked ($\bigstar$) have been deferred to the appendix.

\section{Preliminaries}\label{sec:definitions} \label{sec:preliminaries}
We use~$\mathbb{N}$ to denote the natural numbers, including~$0$. For a positive integer~$n$ and a set~$X$ we use~$\binom{X}{n}$ to denote the collection of all subsets of~$X$ of size~$n$. The \emph{power set} of~$X$ is denoted~$2^{X}$. The set~$\{1, \ldots, n\}$ is abbreviated as~$[n]$. The~$\Oh^*$ notation suppresses polynomial factors in the input size~$n$, such that~$\Oh^*(f(k))$ is shorthand for~$\Oh(f(k) n^{\Oh(1)})$. All our logarithms have base two.

We consider finite, simple, and undirected graphs~$G$, consisting of a vertex set~$V(G)$ and edge set~$E(G) \subseteq \binom{V(G)}{2}$. The neighbors of a vertex~$v$ in~$G$ are denoted~$N_G(v)$. The closed neighborhood of~$v$ is~$N_G[v] := N_G(v) \cup \{v\}$. For a vertex set~$S \subseteq V(G)$ the open neighborhood is~$N_G(S) := \bigcup_{v \in S} N_G(v) \setminus S$ and the closed neighborhood is~$N_G[S] := N_G(S) \cup S$. The subgraph of~$G$ induced by a vertex subset~$U \subseteq V(G)$ is denoted~$G[U]$. The operation of \emph{identifying vertices~$u$ and~$v$} in a graph~$G$ results in the graph~$G'$ that is obtained from~$G$ by replacing the two vertices~$u$ and~$v$ by a new vertex~$w$ with~$N_{G'}(w) = N_G(\{u,v\})$.

An \emph{independent set} is a set of pairwise nonadjacent vertices. A \emph{vertex cover} in a graph~$G$ is a set~$S \subseteq V(G)$ such that~$S$ contains at least one endpoint from every edge. A set~$S \subseteq V(G)$ \emph{dominates} the vertices~$N_G[S]$. A \emph{dominating set} is a vertex set~$S$ such that~$N_G[S] = V(G)$. The associated decision problems ask, given a graph~$G$ and integer~$t$, whether an independent set (dominating set) of size~$t$ exists in~$G$. The size of a maximum independent set (resp.\, minimum dominating set) in~$G$ is denoted~$\optIS(G)$ (resp.~$\optDS(G)$). The \qSAT problem asks whether a given Boolean formula, in conjunctive normal form with clauses of size at most~$q$, has a satisfying assignment.

\begin{seth}[\cite{ImpagliazzoP01,ImpagliazzoPZ01}]
For every $\varepsilon > 0$, there is a constant~$q$ such that \qSAT on $n$ variables cannot be solved in time $\Oh^*((2-\varepsilon)^n)$.
\end{seth}

\subparagraph*{Drawings} 
A \emph{drawing} of a graph~$G$ is a function~$\psi$ that assigns a unique point~$\psi(v) \in \mathbb{R}^2$ to each vertex~$v \in V(G)$, and a curve~$\psi(e) \subseteq \mathbb{R}^2$ to each edge~$e \in E(G)$, such that the following four conditions hold. (1) For~$e = \{u,v\} \in E(G)$, the endpoints of~$\psi(e)$ are exactly~$\psi(u)$ and~$\psi(v)$. (2) The interior of a curve~$\psi(e)$ does not contain the image of any vertex. (3) No three curves representing edges intersect in a common point, except possibly at their endpoints. (4) The interiors of the curves~$\psi(e), \psi(e')$ for distinct edges intersect in at most one point. If the interiors of all the curves representing edges are pairwise-disjoint, then we have a \emph{planar drawing}. In this paper we combine (nonplanar) drawings with crossover gadgets to build planar drawings. A graph is \emph{planar} if it admits a planar drawing.

\subparagraph*{Cutwidth} 
For an $n$-vertex graph~$G$, a \emph{linear layout} of~$G$ is a linear ordering of its vertex set, as given by a bijection~$\pi \colon V(G) \to [n]$. The \emph{cutwidth} of~$G$ with respect to the layout~$\pi$ is:
$$\cutw_\pi(G) = \max_{1 \leq i < n} \bigl | \bigl \{ \{u,v\} \in E(G) \bigmid \pi(u) \leq i \wedge \pi(v) > i \bigr\} \bigr |.$$
The cutwidth~$\cutw(G)$ of a graph~$G$ is the minimum cutwidth attained by any linear layout. It is well-known that~$\cutw(G) \geq \pw(G) \geq \tw(G)$, where the latter denote the pathwidth and treewidth of~$G$, respectively (cf.~\cite{Bodlaender98}).

\section{Planarizing graphs while preserving cutwidth} \label{sect:planarization}
In this section we show how to planarize a graph without blowing up its cutwidth. An intuitive way to think about cutwidth is to consider the vertices as being placed on a horizontal line in the order dictated by the layout~$\pi$, with edges drawn as $x$-monotone curves. For any position~$i$ we consider the gap between vertex~$\pi^{-1}(i)$ and~$\pi^{-1}(i+1)$, and count the edges that \emph{cross} the gap by having one endpoint at position at most~$i$ and the other at position after~$i$. The cutwidth of a layout is the maximum number of edges crossing any single gap; see Figure~\ref{fig:linearlayout}. The simple but useful fact on which our approach hinges is the following. If we obtain~$G'$ by replacing a crossing in the drawing by a new vertex of degree four, and we let~$\pi'$ be the left-to-right order of the vertices in the resulting drawing, then~$\cutw_\pi(G) = \cutw_{\pi'}(G')$. Hence by repeating this procedure we can eliminate all crossings to obtain a planarized version of~$G$ without increasing the cutwidth. To utilize this idea in reductions, we formalize a version of this approach where we planarize the graph by inserting gadgets, rather than simply replacing crossings by degree-four vertices.

\begin{figure}[t]
\centering
\includegraphics{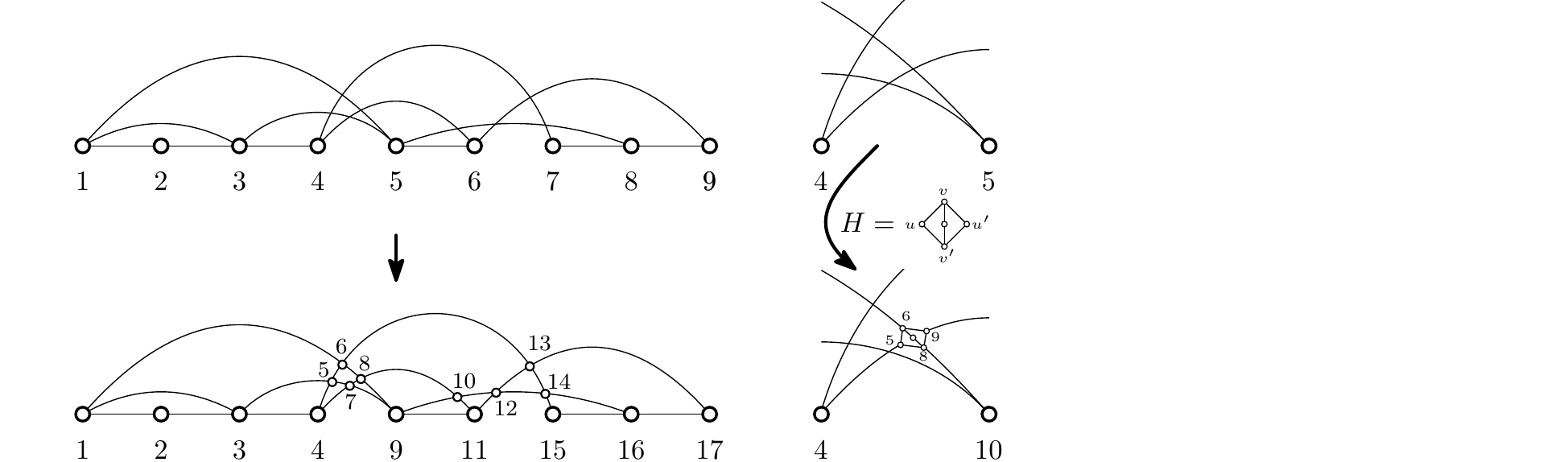}
\caption{Top-left: a linear layout~$\pi$ of a graph~$G$ with~$\cutw_\pi(G) = 4$. The largest cutsize is attained after vertices~$4$ and~$5$. Bottom-left: after inserting vertices at the crossings to obtain~$G'$ and extending~$\pi$ to~$\pi'$ based on the $x$-coordinates of the inserted vertices, we have~$\cutw_\pi(G) = \cutw_{\pi'}(G')$. Top-right: enlarged view. Bottom-right: replacing a crossing by gadget~$H$.}
\label{fig:linearlayout}
\end{figure}

\begin{definition} \label{def:crossover:gadget}
A \emph{crossover gadget} is a graph~$H$ with terminal vertices~$u,u',v,v'$ such that:
\begin{enumerate}
	\item there is a planar drawing~$\psi$ of~$H$ in which all terminals lie on the outer face, and
	\item there is a closed curve intersecting the drawing~$\psi$ only in the terminals, which visits the terminals in the order~$u,v,u',v'$.
\end{enumerate}
\end{definition}

\begin{definition} \label{def:replace}
Let~$\{a,b\}$ and~$\{c,d\}$ be disjoint edges of a graph~$G$, and let~$H$ be a crossover gadget. The operation of \emph{replacing~$\{ \{a,b\}, \{c,d\}\}$ by~$H$} removes edges~$\{a,b\}$ and~$\{c,d\}$, inserts a new copy of the graph~$H$, and inserts the edges~$\{a,u\}, \{u',b\}, \{c,v\}, \{v',d\}$.
\end{definition}

\noindent For a crossover gadget to be useful to planarize instances of a decision problem, a replacement should have a predictable effect on the answer. To formalize this, we say that a \emph{decision problem~$\Pi$ on graphs} is a decision problem whose input consists of a graph~$G$ and integer~$t$.

\begin{definition} \label{def:useful:crossover}
A crossover gadget~$H$ is \emph{useful for a decision problem~$\Pi$ on graphs} if there exists an integer~$c_\Pi$ such that the following holds. If~$(G,t)$ is an instance of~$\Pi$ containing disjoint edges~$\{a,b\}, \{c,d\}$, and~$G'$ is the result of replacing these edges by~$H$, then~$(G,t)$ is a \yes-instance of~$\Pi$ if and only if~$(G', t + c_\Pi)$ is a \yes-instance of~$\Pi$.
\end{definition}

The following theorem proves that a useful crossover gadget can be used to efficiently planarize instances without blowing up their cutwidth.

\begin{theorem} \label{thm:cutwidth:planarization}
If~$H$ is a crossover gadget that is useful for decision problem~$\Pi$ on graphs, then there is a polynomial-time algorithm that, given an instance~$(G,t)$ of~$\Pi$ and a linear layout~$\pi$ of~$G$, outputs an instance~$(G',t')$ and a linear layout~$\pi'$ of~$G'$ such that:
\begin{enumerate}
	\item $G'$ is planar,
	\item $\cutw_{\pi'}(G') \leq \cutw_{\pi}(G) + \cutw(H) + 4$, and
	\item $(G,t)$ is a \yes-instance of~$\Pi$ if and only if~$(G', t')$ is a \yes-instance of~$\Pi$.
\end{enumerate}
\end{theorem}
\begin{proof}
Consider the crossover gadget~$H$ for~$\Pi$ with terminals~$u,u,v,v'$. Let~$\pi_H$ be a linear layout of~$H$ of minimum cutwidth, which is hardcoded into the algorithm along with the integer~$c_\Pi$ described in Definition~\ref{def:useful:crossover}. 

Let~$(G,t)$ be an instance of~$\Pi$ with a linear layout~$\pi$. We start by constructing a (nonplanar) drawing~$\psi$ of~$G$ with the following properties.
\begin{enumerate}
	\item The vertices of~$G$ are placed on the $x$-axis, in the order dictated by~$\pi$.\label{pty:order:pi}
	\item The image of each edge of~$G$ is a strictly $x$-monotone curve.\label{pty:monotone:curves}
	\item If the drawings of two edges intersect in their interior, then their endpoints are all distinct and the corresponding curves properly cross; they do not only touch.\label{pty:cross}
	\item For each pair of edges, their drawings intersect in at most one point.\label{pty:intersect:once}
	\item The $x$-coordinates of all crossings are distinct from each other, and from the $x$-coordinates of the vertices. \label{pty:xorder}
\end{enumerate}
It is easy to see that such a drawing always exists and can be found in polynomial time; we omit the details as they are not interesting. Properties~\ref{pty:order:pi} and~\ref{pty:monotone:curves} together ensure that for any~$i \in [|V(G)|-1]$, the set of edges that cross the gap after vertex~$\pi^{-1}(i)$ in the linear layout is exactly the set of edges intersected by a vertical line between~$\pi^{-1}(i)$ and~$\pi^{-1}(i+1)$, which therefore has size at most~$\cutw_{\pi}(G)$. We will use this property later.

The algorithm replaces the crossings one by one. If two edges~$\{a,b\}$ and~$\{c,d\}$ intersect in their interior, then their endpoints are all distinct by~(\ref{pty:cross}) and they properly cross. Hence we can replace these two edges by a copy of~$H$ as in Definition~\ref{def:replace}. Since there is a planar drawing of~$H$ with the terminals alternating along the outer face, after possibly swapping the labels of~$a$ and~$b$, and of~$c$ and~$d$, the drawing can be updated so that the crossing between~$\{a,b\}$ and~$\{c,d\}$ is eliminated. Since each of~$a,b,c,d$ is made adjacent to exactly one vertex of~$H$, the replacement can be done such that the remaining crossings are in exactly the same locations as before; see the right side of Figure~\ref{fig:linearlayout}. When inserting the crossover gadget, we scale it down sufficiently far that the following holds: all vertices and crossings that were originally on the left of the crossing between~$\{a,b\}$ and~$\{c,d\}$ lie to the left of all vertices that are inserted to replace this crossing; and all vertices and crossings that were on the right of the~$\{\{a,b\},\{c,d\}\}$ crossing, lie to the right of all vertices inserted for its replacement.

Since each pair of edges intersects at most once by~(\ref{pty:intersect:once}), the number of crossings is~$\Oh(|V(G)|^2)$. Hence in polynomial time we can replace all crossings by copies of~$H$ to arrive at a graph~$G'$. If~$\ell$ is the number of replaced crossings, then we set~$t' := t + \ell \cdot c_\Pi$. By Definition~\ref{def:useful:crossover} and transitivity it follows that~$(G,t)$ is a \yes-instance of~$\Pi$ if and only if~$(G',t')$ is a \yes-instance. By construction,~$G'$ is planar. It remains to define a linear layout of~$G'$ and bound its cutwidth.

The layout~$\pi'$ of~$G'$ is defined as follows. Let the \emph{elements} of the original drawing~$\psi$ of~$G$ consist of its vertices and its crossings. The elements of~$\psi$ are linearly ordered by their $x$-coordinates, by~(\ref{pty:xorder}). The linear layout~$\pi'$ of~$G'$ has one block per element of~$\psi$, and these blocks are ordered according to the $x$-coordinates of the corresponding element. For elements that consist of a vertex~$v$, the block simply consists of~$v$. For elements that consist of a crossing~$X$, the block consists of the vertices of the copy of~$H$ that was inserted to replace~$X$, in the order dictated by~$\pi_H$. It is easy to see that~$\pi'$ can be constructed in polynomial time.

We classify the edges of~$G'$ into two types. We have \emph{internal} edges, which are edges within an inserted copy of~$H$, and we have \emph{external edges} which connect two different copies of~$H$, or which connect a vertex of~$V(G) \cap V(G')$ to a copy of~$H$. Using this classification we argue that for an arbitrary vertex~$v^*$ of~$G'$, the cut crossing the gap after vertex~$v^*$ in~$\pi'$ contains at most~$\cutw_{\pi}(G) + \cutw(H) + 4$ edges. To do so, we distinguish two cases depending on whether~$v^*$ is an original vertex from~$G$, or was inserted as part of a copy of~$H$.

\begin{claim}
If~$v^* \in V(G) \cap V(G')$, then the size of the cut after~$v^*$ in~$\pi'$ is at most~$\cutw_{\pi}(G)$.
\end{claim}
\begin{claimproof}
The layout~$\pi'$ consists of blocks, and~$v^* \in V(G) \cap V(G')$ is a block. So for each copy~$C$ of a crossover gadget, the vertices of~$C$ all appear on the same side of~$v^*$ in the ordering. Hence no internal edge of~$C$ crosses the cut after~$v^*$, implying that no internal edge is in the cut. Each external edge crossing the cut is (a segment of) an edge of~$G$ that is intersected by a vertical line after~$v^*$ in the drawing~$\psi$; see Figure~\ref{fig:linearlayout}. As such a line intersects at most~$\cutw_{\pi}(G)$ edges as observed above, the cut after~$v^*$ has size at most~$\cutw_{\pi}(G)$.
\end{claimproof}

\begin{claim}
If~$v^* \in V(G')$ is a vertex of a copy~$C$ of a crossover gadget that was inserted to replace a crossing~$X$, then the size of the cut after~$v^*$ in~$\pi'$ is at most~$\cutw_\pi(G) + \cutw(H) + 4$.
\end{claim}
\begin{claimproof}
The number of internal edges in the cut after~$v^*$ is at most~$\cutw(H)$, since the only internal edges in the cut all belong to the same copy~$C$ that contains~$v^*$ and we ordered them according to an optimal layout~$\pi_H$. There are at most four external edges incident on a vertex of~$C$, which contribute at most four to the cut. Finally, for each of the remaining external edges in the cut there is a unique edge of~$G$ intersected by a vertical line through crossing~$X$ in the drawing~$\psi$. As at most~$\cutw_{\pi}(G)$ edges are intersected by any vertical line, as observed above, it follows that the size of the cut is at most~$\cutw_\pi(G) + \cutw(H) + 4$.
\end{claimproof}

The two claims together show that any gap in the ordering~$\pi'$ is crossed by at most~$\cutw_\pi(G) + \cutw(H) + 4$ edges, which bounds the cutwidth of~$G'$ as required.
\end{proof}

Using Theorem~\ref{thm:cutwidth:planarization} we can now elegantly prove Theorem~\ref{thm:independentset:lb} by combining two known results. 

{
\renewcommand{\thetheorem}{\ref{thm:independentset:lb}}
\begin{theorem}
\indsetthm
\end{theorem}
\addtocounter{theorem}{-1}
}
\begin{proof}
First, we observe that the crossover gadget for \textsc{Vertex Cover} due to Garey, Johnson, and Stockmeyer~\cite[Thm. 2.7]{GareyJS76} satisfies our conditions of a useful crossover gadget. Since an $n$-vertex graph has a vertex cover of size~$k$ if and only if it has an independent set of size~$n-k$, it also acts as a useful crossover gadget for \textsc{Independent Set} with~$c_\pi = 9$ (cf.~\cite[Proposition 20]{JansenW16}). Second, we observe that by a different analysis of a construction due to Lokshtanov et al.~\cite{LokshtanovMS11}, it follows that (assuming SETH) there is no~$\varepsilon > 0$ such that \textsc{Independent Set} on a graph~$G$ with a linear layout of cutwidth~$k$ can be solved in time~$\Oh^*((2-\varepsilon)^k)$. We prove this in Theorem~\ref{thm:ind:set:cutwidth:lb} in the appendix. By Theorem~\ref{thm:cutwidth:planarization}, if such a runtime could be achieved on \emph{planar} graphs of a given cutwidth, it could be achieved for a general graph as well, since the insertion of crossover gadgets increases the cutwidth by only a constant. Hence Theorem~\ref{thm:independentset:lb} follows.
\end{proof}

\section{Lower bound for dominating set on planar graphs of bounded cutwidth}\label{sec:lowds}

In this section we prove a runtime lower bound for solving \textsc{Dominating Set} on planar graphs of bounded cutwidth. Our starting point is the insight that through a minor modification, the lower bound by Lokshtanov et al.~\cite{LokshtanovMS11} for the parameterization by pathwidth can be lifted to apply to the parameterization by cutwidth as well.

\newcommand{\nonplanardomsetthm}{Assuming SETH, there is no~$\varepsilon > 0$ such that \textsc{Dominating Set} on a (nonplanar) graph~$G$ given along with a linear layout of cutwidth~$k$ can be solved in time~$\Oh^*((3-\varepsilon)^k)$.}
\begin{theorem}[$\bigstar$] \label{thm:domset:nonplanar:lb}
\nonplanardomsetthm
\end{theorem}

\begin{figure}[t]
	\centering
	\includegraphics{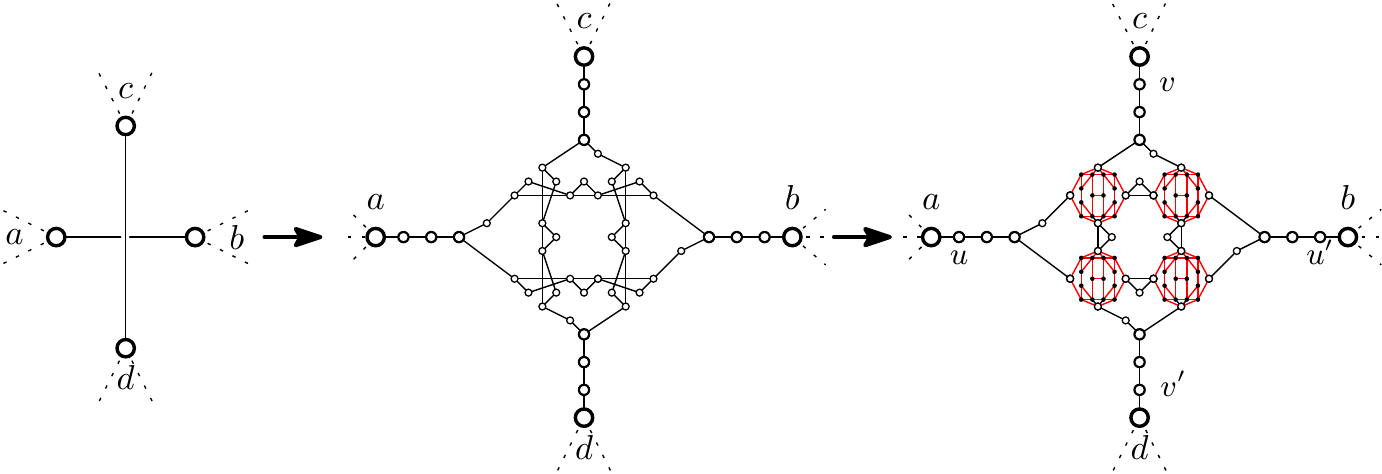}
	\caption{Overview of the method to eliminate an edge crossing in an instance of \textsc{Dominating Set}. Each edge is transformed into a double-path structure, to turn a single crossing edge into four crossing triangles (middle figure). Then each crossing triangle is replaced by a planar gadget (right figure). Some vertices have been omitted for readability. For each red edge~$\{u,v\}$, there is a hidden degree-two vertex in the graph that forms a triangle with~$u$ and~$v$.}
	\label{fig:dominatingset:planarization}
\end{figure}

Our contribution is to extend the lower bound of Theorem~\ref{thm:domset:nonplanar:lb} to apply to \emph{planar} graphs. Following the strategy outlined in Section~\ref{sect:planarization}, to achieve this it suffices to develop a useful crossover gadget for \textsc{Dominating Set} as per Definition~\ref{def:useful:crossover}. Since our crossover gadget is fairly complicated (it has more than 100 vertices), we describe its design in steps. The main idea is as follows. We first show that an edge in a \textsc{Dominating Set} instance can be replaced by a longer double-path structure, which contains several triangles. Then we show that when two triangles cross, we can replace their crossing by a suitable adaptation of the \textsc{Vertex Cover} crossover gadget due to Garey, Johnson, and Stockmeyer~\cite[Thm. 2.7]{GareyJS76}. This two-step approach is illustrated in Figure~\ref{fig:dominatingset:planarization}. We follow the same two steps in proving its correctness, starting with the insertion of the double-path structure.

\newcommand{\doublepathlemma}{Let~$\{x,y\}$ be an edge in a graph~$G$. If~$G'$ is obtained from~$G$ by replacing $\{x,y\}$ by a double-path structure as shown in Figure~\ref{fig:dominatingset:gadget}, then $\optDS(G') = \optDS(G) + 6$.}
\begin{lemma} \label{lemma:insert:doublepath}
\doublepathlemma
\end{lemma}
\begin{proof}
	We prove equality by establishing matching upper- and lower bounds.
	
	(\textbf{$\leq$}) To show~$\optDS(G') \leq \optDS(G) + 6$, consider a minimum dominating set~$S \subseteq V(G)$ of~$G$. If~$S \cap \{x,y\} = \emptyset$, or~$S \cap \{x,y\} = \{x,y\}$, then the edge~$\{x,y\}$ is not used to dominate vertex~$x$ or~$y$, and therefore~$S \cup \{b_x, b_y, e_x, e_y, g_x, g_y\}$ is a dominating set of size~$|S|+6 = \optDS(G) + 6$ in graph~$G'$; see Figure~\ref{fig:dominatingset:gadget}. If~$S \cap \{x,y\} = \{x\}$, then~$S \cup \{c_x, f_x, h_x, e_y, g_y, a_y\}$ is a dominating set of size~$\optDS(G) + 6$ in~$G'$: the vertex~$a_y$ takes over the role of dominating~$y$ after the direct edge~$\{x,y\}$ is removed, while~$a_x$ is dominated from~$x$. Symmetrically, if~$S \cap \{x,y\} = \{y\}$, then~$S \cup \{c_y, f_y, h_y, e_x, g_x, a_x\}$ is a dominating set in~$G'$ of size~$\optDS(G) + 6$.

	(\textbf{$\geq$}) To show~$\optDS(G') \geq \optDS(G) + 6$, we instead prove~$\optDS(G) \leq \optDS(G') - 6$. Consider a minimum dominating set~$S' \subseteq V(G')$ of~$G'$. Let~$B$ be the vertices in the interior of the double-path structure that was inserted into~$G'$ to replace edge~$\{x,y\}$. If~$|S' \cap B| \geq 7$, then~$(S \setminus B) \cup \{x\}$ is a dominating set of size at most~$|S'| - 6 \leq \optDS(G') - 6$ in~$G$, since~$x$ dominates itself and~$y$ using the edge~$\{x,y\}$. We assume~$|S' \cap B| \leq 6$ in the remainder. Then we have~$|S' \cap B| = 6$: the closed neighborhoods of the six vertices~$\{b_x, b_y, t_x, t_y, t''_x, t''_y\}$ are contained entirely within~$B$, and are pairwise disjoint. Hence these six vertices must be dominated by six distinct vertices from~$B$. If~$S' \cap \{a_x, a_y\} = \emptyset$ then the vertices~$x$ and~$y$ are not dominated from within the double-path structure, implying that~$S' \setminus B$ is a dominating set in~$G$ of size~$|S'| - 6 \leq \optDS(G') - 6$. It remains to consider the case that~$S'$ contains~$a_x$, or~$a_y$, or both. 

	\begin{claim}
		Let~$B' \subseteq B$ be a set of size six that dominates the vertices~$B \setminus \{a_x, a_y\}$. If~$B'$ contains~$a_x$, then~$B'$ does not dominate~$a_y$. Analogously, if~$B'$ contains~$a_y$ then it does not dominate~$a_x$.
	\end{claim}
	\begin{claimproof}
		We prove that if~$a_y \in B'$, then~$B'$ does not dominate~$a_x$. The other statement follows by symmetry. So assume for a contradiction that~$B'$ contains~$a_y$ and dominates~$a_x$, which implies it contains~$a_x$ or~$b_x$. Since~$B'$ dominates the interior of the double-path structure, it contains at least one vertex from the closed neighborhoods of~$t_x, t_y, t''_x, t''_y$. Since these are pairwise disjoint, and do not contain~$a_y$,~$a_x$, or~$b_x$, the set~$B'$ contains a vertex from the closed neighborhood of each of~$\{t_x, t_y, t''_x, t''_y\}$. Since~$B'$ has size six, besides the four vertices from these closed neighborhoods, the vertex~$a_y$, and the one vertex in~$\{a_x, b_x\}$ there can be no further vertices in~$B'$. Hence~$B'$ does not contain~$c_x$ or~$d_x$, as these do not occur in the stated closed neighborhoods. This implies that to dominate~$d_x$, the set~$B'$ contains~$e_x$. Then vertex~$t'_x$ is not dominated by the vertex from~$N_{G'}[t_x]$, and must therefore be dominated by the vertex in~$B'$ from~$N_{G'}[t''_x]$, implying that~$g_x \in B'$. But the vertices mentioned so far do not dominate~$c_y$, and regardless of how a vertex is chosen from the closed neighborhoods of~$t_y$ and~$t''_y$, the resulting choice does not dominate~$c_y$ since no vertex from the closed neighborhoods of~$t_y, t''_y$ is adjacent to~$c_y$. So~$c_y$ is not dominated by~$B'$; a contradiction.
	\end{claimproof}

\begin{figure}[t]
	\centering
	\includegraphics{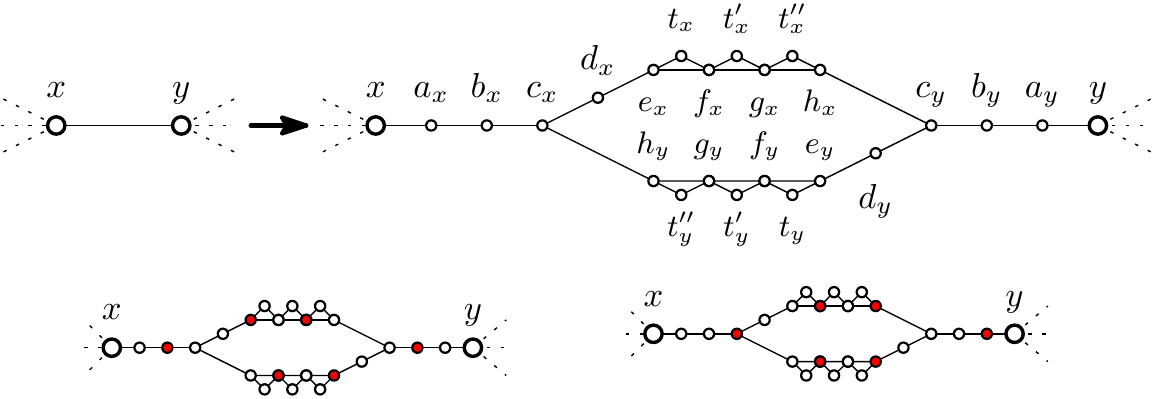}
	\caption{The double-path structure for \textsc{Dominating Set} is the subgraph on the top right minus the vertices~$x$ and~$y$. Top: an edge~$\{x,y\}$ is replaced by a double-path structure. Bottom-left: the interior of the double-path structure can be dominated by six vertices (in red). Bottom-right: there is a set of six vertices that dominates~$y$ and all vertices of the double-path structure except~$a_x$.}
	\label{fig:dominatingset:gadget}
\end{figure}

	Using the claim we finish the proof. The set~$B' := S \cap B$ has size six and dominates all of~$B \setminus \{a_x, a_y\}$, since those vertices cannot be dominated from elsewhere. If~$B'$ contains~$a_x$, then~$B'$ does not dominate~$a_y$. Since~$S'$ is a dominating set, and~$y$ is the only neighbor of~$a_y$ outside of~$B$, it follows that~$y \in S'$. But then~$S' \setminus B$ is a dominating set in~$G$ of size~$|S'| - 6 = \optDS(G') - 6$ in~$G$: the edge~$\{x,y\}$ in~$G$ ensures that~$y$ dominates~$x$. If~$B'$ contains~$a_y$ instead, then the symmetric argument applies. Hence~$\optDS(G) \leq \optDS(G') - 6$.
\end{proof}

Using Lemma~\ref{lemma:insert:doublepath}, we can replace a direct edge by a double-path structure while controlling the domination number. This allows two crossing edges to be reduced to four crossing triangles as in Figure~\ref{fig:dominatingset:planarization}. Even though more crossings are created in this way, these crossing triangles actually help to planarize the graph. The key point is that crossing triangles enforce a dominating set to locally act like a vertex cover, which allows us to exploit a known gadget for \textsc{Vertex Cover}. The following two statements are useful to formalize these ideas. Recall that a vertex~$v$ is \emph{simplicial} in a graph~$G$ if~$N_G(v)$ forms a clique.

\begin{observation} \label{obs:domset:avoid:degtwo}
If~$I$ is an independent set of simplicial degree-two vertices in a graph~$G$, then~$G$ has a minimum dominating set that contains no vertex of~$I$.
\end{observation}

\begin{proposition} \label{proposition:domset:is:vc}
	Let~$U$ be a set of vertices in a graph~$G$, such that for each edge~$\{x,y\} \in E(G[U])$ there is a vertex~$v \in V(G) \setminus U$ with~$N_G(v) = \{x,y\}$. Then there is a minimum dominating set~$S$ of~$G$ such that~$S$ forms a vertex cover of~$G[U]$.
\end{proposition}
\begin{proof}
	Construct a set~$I$ as follows. For each~$\{x,y\} \in E(G[U])$, add a vertex~$v \in V(G) \setminus U$ with~$N_G(v) = \{x,y\}$ to~$I$. Then~$I$ is an independent set of simplicial degree-two vertices. By Observation~\ref{obs:domset:avoid:degtwo} there is a minimum dominating set~$S$ of~$G$ that contains no vertex of~$I$. Then~$S \cap U$ is a vertex cover of~$G[U]$: for an arbitrary edge~$\{x,y\} \in E(G[U])$ there is a vertex~$v$ in~$I$ whose open neighborhood is~$\{x,y\}$. Since~$I \cap S = \emptyset$, at least one of~$x$ and~$y$ belongs to~$S$ to dominate~$v$. Hence the edge~$\{x,y\}$ is covered by~$S$.
\end{proof}

Proposition~\ref{proposition:domset:is:vc} relates minimum dominating sets to vertex covers. We therefore use a simplified version of a \textsc{Vertex Cover} crossover gadget in our design. We exploit the graph~$\Hvc$ with four terminals~$\{x,y,p,q\}$ that is shown in Figure~\ref{fig:vertexcover:gadget}. It was obtained by applying the ``folding'' reduction rule for \textsc{Vertex Cover}~\cite[Lemma 2.3]{ChenKJ01} on the gadget by Garey et al.~\cite{GareyJS76} and omitting two superfluous edges. We use the following property of the graph~$\Hvc$. It states that in \Hvc, for every axis from which a vertex cover contains no terminal vertex, the number of non-terminal vertices used in a vertex cover increases.

\begin{proposition} \label{prop:axes:hvc}
	Let $S$ be a vertex cover of \Hvc and let~$\ell \in \{0,1,2\}$ be the number of pairs among~$\{p,q\}$ and~$\{x,y\}$ from which~$S$ contains no vertices. Then~$|S \setminus \{p,q,x,y\}| \geq 9 + \ell$.
\end{proposition}
\begin{proof}
We first show~$|S \setminus \{p,q,x,y\}| \geq 9$ for any vertex cover~$S$ of~$\Hvc$, proving the claim for~$\ell=0$. The non-terminal vertices of~$\Hvc$ can be partitioned into four vertex-disjoint triangles and an edge that is vertex-disjoint from the triangles. From any triangle, a vertex cover contains at least two vertices. From the remaining edge, it contains at least one vertex.
	
	If~$S$ contains no vertex of~$\{p,q\}$, then as illustrated in the middle of Figure~\ref{fig:vertexcover:gadget},~$S$ contains at least eleven non-terminals. Hence~$|S \setminus \{p,q,x,y\}| \geq 11 \geq 9 + \ell$.
	
	If the previous case does not apply, then~$\ell \leq 1$ since~$S$ contains a vertex of~$\{p,q\}$. If~$S$ contains no vertex of~$\{x,y\}$, then as illustrated on the right of Figure~\ref{fig:vertexcover:gadget},~$S$ contains at least ten non-terminals. Hence~$|S \setminus \{p,q,x,y\}| \geq 10 \geq 9 + \ell$.
\end{proof}

Using Proposition~\ref{prop:axes:hvc} we prove that replacing two crossing triangles in a \textsc{Dominating Set} instance by the gadget, increases the optimum by exactly nine.

\begin{figure}[t]
	\centering
	\includegraphics{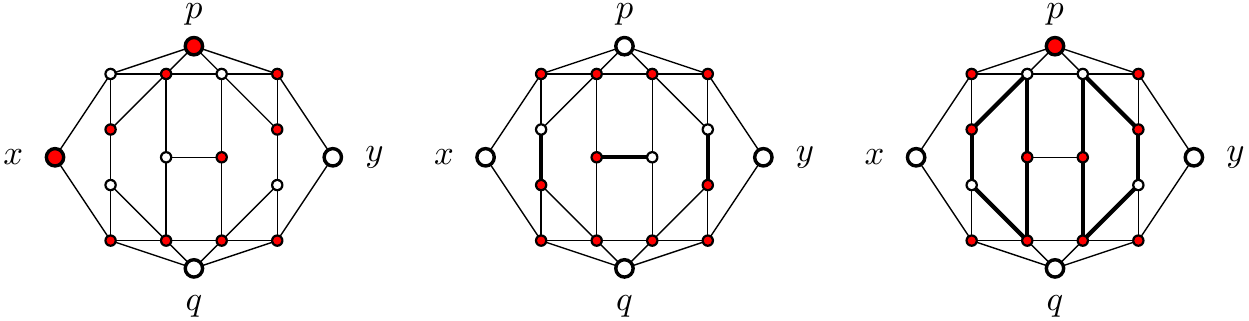}
	\caption{Three copies of the 18-vertex gadget graph~$\Hvc$, which has four terminals~$\{x,y,p,q\}$. Left: A vertex cover for~$\Hvc$ that contains~$p$ and~$q$ and has size eleven is shown in red. Middle: Any vertex cover for~$\Hvc$ that does not contain~$p$ or~$q$ contains the neighbors of~$p$ and~$q$ and at least one endpoint of the three thick edges, and contains at least eleven non-terminals. Right: Any vertex cover for~$\Hvc$ that does not contain~$x$ or~$y$ contains the four neighbors of~$x$ and~$y$ and at least three vertices from each of the two highlighted five-cycles, and contains at least ten non-terminals.}
	\label{fig:vertexcover:gadget}
\end{figure}

\begin{lemma} \label{lemma:replace:triangles}
	Let~$G$ be a graph, and let~$\{x,y,z\}$ and~$\{p,q,r\}$ be two vertex-disjoint triangles in~$G$ such that~$z$ and~$r$ have degree two in~$G$. Then the graph~$G'$ obtained from~$G$ by replacing~$z$ and~$r$ by \Hvc as in Figure~\ref{fig:dominatingset:triangle:gadget} satisfies~$\optDS(G') = \optDS(G) + 9$.
\end{lemma}
\begin{proof}
	We prove equality by establishing matching upper- and lower bounds.

	(\textbf{$\leq$}) Consider a minimum dominating set~$S$ in~$G$ that does not contain~$r$ or~$z$, which exists by Observation~\ref{obs:domset:avoid:degtwo}. Then~$S$ contains at least one of~$\{p,q\}$ to dominate~$z$, and at least one of~$\{x,y\}$ to dominate~$r$. We assume without loss of generality, by symmetry, that~$p \in S$ and~$x \in S$. As shown in Figure~\ref{fig:vertexcover:gadget}, there is a vertex cover for~$\Hvc$ of size~$11$ that contains~$p$ and~$x$, and therefore contains nine vertices from the interior of~$\Hvc$. Let~$T$ be this set of nine vertices, and note that~$T$ includes a neighbor of~$q$ and a neighbor of~$y$. We claim that~$S' := S \cup T$ is a dominating set for~$G'$ of size~$|S| + 9 \leq \optDS(G)+9$. Since~$T \cup \{p,x\}$ is a vertex cover of~$\Hvc$ and~$\Hvc$ has no isolated vertices, each vertex of~$\Hvc$ has a neighbor in~$S'$ and is dominated. The degree-two vertices that are inserted into~$G'$ in the last step are dominated by the vertex that covers the edge with which they form a triangle. Vertices~$q$ and~$y$ are dominated from their neighbors in~$T$. Finally, the remaining vertices of~$G'$ are dominated in the same way as in~$G$.
	
	(\textbf{$\geq$}) To prove~$\optDS(G') \geq \optDS(G) + 9$, we instead show~$\optDS(G) \leq \optDS(G') - 9$. Let~$U \subseteq V(G')$ denote the vertices from the copy of~$\Hvc$ that was inserted;~$U$ contains~$p,q,x,y$, but~$U$ does not contain the degree-two vertices that were inserted as the last step of the transformation. By Proposition~\ref{proposition:domset:is:vc}, there is a minimum dominating set~$S'$ of~$G'$ such that~$S' \cap U$ is a vertex cover of~$G'[U]$. Let~$\ell \in \{0,1,2\}$ be the number of pairs among~$\{p,q\}$ and~$\{x,y\}$ from which~$S'$ contains no vertices. Since~$S' \cap U$ is a vertex cover of~$G'[U]$, which is isomorphic to~$\Hvc$, by Proposition~\ref{prop:axes:hvc} we know that~$|(S' \cap U) \setminus \{p,q,x,y\}| \geq 9 + \ell$. Now let~$S$ be obtained from~$S'$ by removing all vertices of~$(S' \cap U) \setminus \{p,q,x,y\}$, adding vertex~$x$ if~$S' \cap \{x,y\} = \emptyset$, and adding vertex~$y$ if~$S' \cap \{p,q\} = \emptyset$. Then~$|S| \leq |S'| - 9$ since we remove~$9+\ell$ vertices and add~$\ell$ new ones. Since~$S$ contains at least one vertex from~$\{p,q\}$ and at least one vertex from~$\{x,y\}$, it dominates the two triangles in~$G$. Since it contains a superset of the terminal vertices that~$S'$ contains, the remaining vertices of the graph are dominated as before. Hence~$S$ is a dominating set in~$G$ and~$\optDS(G) \leq \optDS(G') - 9$.
\end{proof}

\begin{figure}[t]
	\centering
	\includegraphics{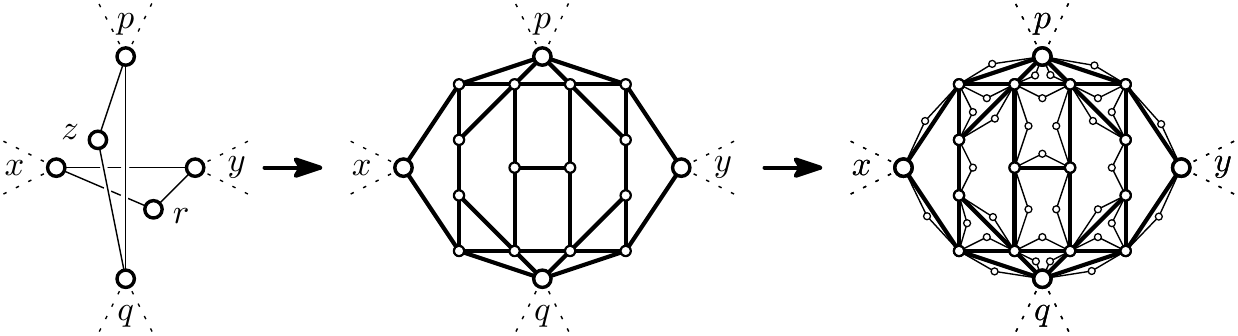}
	\caption{Illustration of how a crossing between two triangles is eliminated in an instance of \textsc{Dominating Set}. A copy of~$\Hvc$ is inserted, whose four terminals are identified with the four endpoints of the crossing edge. For each edge~$\{u,v\}$ of the inserted copy of~$\Hvc$, an additional degree-two vertex is inserted that forms a triangle with~$u$ and~$v$.}
	\label{fig:dominatingset:triangle:gadget}
\end{figure}

Using the material so far, we can prove that the transformation operation in Figure~\ref{fig:dominatingset:planarization} increases the size of an optimal dominating set by exactly~$48$.

\begin{lemma} \label{lemma:planarizedomset:deltaopt}
	Let~$\{a,b\}$ and~$\{c,d\}$ be two disjoint edges of a graph~$G$. Let~$G'$ be the graph obtained by replacing these two edges as in Figure~\ref{fig:dominatingset:planarization}. Then~$\optDS(G') = \optDS(G) + 48$.
\end{lemma}
\begin{proof}
	The transformation depicted in Figure~\ref{fig:dominatingset:planarization} can be broken down into six steps: transform~$\{a,b\}$ into a double-path structure, transform~$\{c,d\}$ into a double-path structure, and perform four operations in which crossing triangles are replaced by gadgets. By Lemma~\ref{lemma:insert:doublepath}, the two double-path insertions increase the size of a minimum dominating set by exactly~$2\cdot 6$. By Lemma~\ref{lemma:replace:triangles}, the four steps in which crossing triangles are eliminated increase the size of a minimum domination set by exactly~$4 \cdot 9$. Hence~$\optDS(G') = \optDS(G) + 12+36$.
\end{proof}

Using Lemma~\ref{lemma:planarizedomset:deltaopt} we easily obtain the following.

\begin{lemma} \label{lemma:useful:crossover:domset}
	There is a useful crossover gadget for \textsc{Dominating Set}.
\end{lemma}
\begin{proof}
	The gadget that is inserted to replace two edges~$\{a,b\}$ and~$\{c,d\}$ in the procedure of Figure~\ref{fig:dominatingset:planarization} is planar and has its terminals~$u,u',v,v'$ on the outer face in the appropriate cyclic ordering. Since Lemma~\ref{lemma:planarizedomset:deltaopt} shows that the replacement increases the size of a minimum dominating set by exactly 48, it follows that the structure serves as a useful crossover gadget for \textsc{Dominating Set} as per Definition~\ref{def:useful:crossover}.
\end{proof}

Theorem~\ref{thm:domset:planar:lb} now follows by combining Lemma~\ref{lemma:useful:crossover:domset} with the planarization argument of Theorem~\ref{thm:cutwidth:planarization} and the lower bound for the nonplanar case given by Theorem~\ref{thm:domset:nonplanar:lb}.

{
\renewcommand{\thetheorem}{\ref{thm:domset:planar:lb}}
\begin{theorem}
\domsetthm
\end{theorem}
\addtocounter{theorem}{-1}
}
\begin{proof}
	Suppose \textsc{Dominating Set} on a planar graph with a given linear layout of cutwidth~$k$ can be solved in~$\Oh^*((3-\varepsilon)^k)$ time for some~$\varepsilon > 0$, by an algorithm called~$A$. Then \textsc{Dominating Set} on a nonplanar graph with a given layout of cutwidth~$k$ can be solved in time~$\Oh^*((3-\varepsilon)^k)$ by reducing it to a planar graph with a linear layout of cutwidth~$k + \Oh(1)$ (using Theorem~\ref{thm:cutwidth:planarization} and the existence of a useful crossover gadget; this blows up the graph size by at most a polynomial factor) and then running~$A$. By Theorem~\ref{thm:domset:nonplanar:lb}, this contradicts SETH.
\end{proof}

\section{Conclusion}\label{sec:conc}
In this work we have investigated whether SETH-based lower bounds for solving problems on graphs of bounded treewidth also apply for (1) planar graphs and (2) graphs of bounded cutwidth. To answer these questions, we showed that the graph parameter cutwidth can be preserved when reducing to a planar instance using suitably restricted crossover gadgets. 

For both problems considered in this work, the runtime lower bound for solving the problem on graphs of bounded cutwidth continues to hold for \emph{planar} graphs of bounded cutwidth. Hence planarity seems to offer no algorithmic advantage when working with graphs of bounded cutwidth. Moreover, for both \textsc{Independent Set} and \textsc{Dominating Set} the runtime lower bound for the treewidth parameterization also applies for cutwidth.

Future work may explore other combinatorial problems on graphs of bounded cutwidth. For example, what is the optimal running time for \textsc{Feedback Vertex Set}, \textsc{Odd Cycle Transversal}, or \textsc{Hamiltonian Cycle} on graphs of bounded cutwidth? What is the complexity of the cutwidth parameterization of these problems on planar graphs? For the \textsc{Graph $q$-Coloring} problem, these questions are answered in a recent manuscript by an overlapping set of authors~\cite{JansenN17}: planarity offers no advantage, but the parameterization by cutwidth~$k$ can be solved in time~$\Oh^*(2^k)$ for all~$q$, sharply contrasting that the treewidth parameterization cannot be solved in time~$\Oh^*((q-\varepsilon)^k)$ under SETH.

%Open problems: look at other problems parameterized by cutwidth. What about the slightly superexponential problems~\cite{LokshtanovMS11a}? Dynamic programming on sphere-cut decompositions probably speeds them up. Are there crossover gadgets for \textsc{Feedback Vertex Set} or \textsc{Odd Cycle Transversal}? What is the best exponent that can be obtained for solving \textsc{Partition into Triangles} on graphs of cutwidth~$k$? We expect it to be solvable in~$\Oh^*(c^k)$ time for some~$c\leq 2^{2/3}$.

%What about DP on graphs of bounded bandwidth? The Independent Set lower bound for nonplanar graphs can be salvaged, but for planar graphs it seems to break, and for other problems requiring more involved gadgets on the repeater line it also seems to break.

\bibliography{Cutwidth}

\begin{thebibliography}{10}

\bibitem{BasteS15}
Julien Baste and Ignasi Sau.
\newblock The role of planarity in connectivity problems parameterized by
  treewidth.
\newblock {\em Theor. Comput. Sci.}, 570:1--14, 2015.
\newblock \href {http://dx.doi.org/10.1016/j.tcs.2014.12.010}
  {\path{doi:10.1016/j.tcs.2014.12.010}}.

\bibitem{BjorklundHKK07}
Andreas Bj{\"{o}}rklund, Thore Husfeldt, Petteri Kaski, and Mikko Koivisto.
\newblock Fourier meets {M}{\"{o}}bius: {F}ast subset convolution.
\newblock In David~S. Johnson and Uriel Feige, editors, {\em Proc. 39th STOC},
  pages 67--74. {ACM}, 2007.
\newblock \href {http://dx.doi.org/10.1145/1250790.1250801}
  {\path{doi:10.1145/1250790.1250801}}.

\bibitem{Bodlaender98}
Hans~L. Bodlaender.
\newblock A partial $k$-arboretum of graphs with bounded treewidth.
\newblock {\em Theor. Comput. Sci.}, 209(1-2):1--45, 1998.
\newblock \href {http://dx.doi.org/10.1016/S0304-3975(97)00228-4}
  {\path{doi:10.1016/S0304-3975(97)00228-4}}.

\bibitem{BodlaenderCKN15}
Hans~L. Bodlaender, Marek Cygan, Stefan Kratsch, and Jesper Nederlof.
\newblock Deterministic single exponential time algorithms for connectivity
  problems parameterized by treewidth.
\newblock {\em Inf. Comput.}, 243:86--111, 2015.
\newblock \href {http://dx.doi.org/10.1016/j.ic.2014.12.008}
  {\path{doi:10.1016/j.ic.2014.12.008}}.

\bibitem{BodlaenderK08}
Hans~L. Bodlaender and Arie M. C.~A. Koster.
\newblock Combinatorial optimization on graphs of bounded treewidth.
\newblock {\em Comput. J.}, 51(3):255--269, 2008.
\newblock \href {http://dx.doi.org/10.1093/comjnl/bxm037}
  {\path{doi:10.1093/comjnl/bxm037}}.

\bibitem{ChenKJ01}
Jianer Chen, Iyad~A. Kanj, and Weijia Jia.
\newblock Vertex cover: Further observations and further improvements.
\newblock {\em J. Algorithms}, 41(2):280--301, 2001.
\newblock \href {http://dx.doi.org/10.1006/jagm.2001.1186}
  {\path{doi:10.1006/jagm.2001.1186}}.

\bibitem{Courcelle90}
Bruno Courcelle.
\newblock The monadic second-order logic of graphs {I}: Recognizable sets of
  finite graphs.
\newblock {\em Inf. Comput.}, 85(1):12--75, 1990.
\newblock \href {http://dx.doi.org/10.1016/0890-5401(90)90043-H}
  {\path{doi:10.1016/0890-5401(90)90043-H}}.

\bibitem{CyganKN13}
Marek Cygan, Stefan Kratsch, and Jesper Nederlof.
\newblock Fast hamiltonicity checking via bases of perfect matchings.
\newblock In Dan Boneh, Tim Roughgarden, and Joan Feigenbaum, editors, {\em
  Proc. 45th STOC}, pages 301--310. {ACM}, 2013.
\newblock \href {http://dx.doi.org/10.1145/2488608.2488646}
  {\path{doi:10.1145/2488608.2488646}}.

\bibitem{CyganNPPRW11}
Marek Cygan, Jesper Nederlof, Marcin Pilipczuk, Michal Pilipczuk, Johan M.~M.
  van Rooij, and Jakub~Onufry Wojtaszczyk.
\newblock Solving connectivity problems parameterized by treewidth in single
  exponential time.
\newblock In {\em Proc. 52nd FOCS}, pages 150--159, 2011.
\newblock \href {http://dx.doi.org/10.1109/FOCS.2011.23}
  {\path{doi:10.1109/FOCS.2011.23}}.

\bibitem{DemaineH08}
Erik~D. Demaine and MohammadTaghi Hajiaghayi.
\newblock The bidimensionality theory and its algorithmic applications.
\newblock {\em Comput. J.}, 51(3):292--302, 2008.
\newblock \href {http://dx.doi.org/10.1093/comjnl/bxm033}
  {\path{doi:10.1093/comjnl/bxm033}}.

\bibitem{DiazPS02}
Josep D{\'{\i}}az, Jordi Petit, and Maria~J. Serna.
\newblock A survey of graph layout problems.
\newblock {\em {ACM} Comput. Surv.}, 34(3):313--356, 2002.
\newblock \href {http://dx.doi.org/10.1145/568522.568523}
  {\path{doi:10.1145/568522.568523}}.

\bibitem{Eppstein16}
David Eppstein.
\newblock {P}athwidth of planarized drawing of~{$K_{3,n}$}.
\newblock {TheoryCS} {S}tack{E}xchange question, 2016.
\newblock URL: \url{http://cstheory.stackexchange.com/questions/35974/}.

\bibitem{Eppstein17a}
David Eppstein.
\newblock The effect of planarization on width.
\newblock In {\em Proc. 25th GD}, volume 10692 of {\em LNCS}, pages 560--572,
  2017.
\newblock \href {http://dx.doi.org/10.1007/978-3-319-73915-1_43}
  {\path{doi:10.1007/978-3-319-73915-1_43}}.

\bibitem{FominLPS16}
Fedor~V. Fomin, Daniel Lokshtanov, Fahad Panolan, and Saket Saurabh.
\newblock Efficient computation of representative families with applications in
  parameterized and exact algorithms.
\newblock {\em J. {ACM}}, 63(4):29:1--29:60, 2016.
\newblock \href {http://dx.doi.org/10.1145/2886094}
  {\path{doi:10.1145/2886094}}.

\bibitem{GareyJS76}
M.R. Garey, D.S. Johnson, and L.~Stockmeyer.
\newblock Some simplified {NP}-complete graph problems.
\newblock {\em Theoretical Computer Science}, 1(3):237--267, 1976.
\newblock \href {http://dx.doi.org/10.1016/0304-3975(76)90059-1}
  {\path{doi:10.1016/0304-3975(76)90059-1}}.

\bibitem{GiannopoulouPRT16}
Archontia~C. Giannopoulou, Michal Pilipczuk, Jean{-}Florent Raymond,
  Dimitrios~M. Thilikos, and Marcin Wrochna.
\newblock Cutwidth: Obstructions and algorithmic aspects.
\newblock In {\em Proc. 11th IPEC}, volume~63 of {\em LIPIcs}, pages
  15:1--15:13, 2016.
\newblock \href {http://dx.doi.org/10.4230/LIPIcs.IPEC.2016.15}
  {\path{doi:10.4230/LIPIcs.IPEC.2016.15}}.

\bibitem{ImpagliazzoP01}
Russel Impagliazzo and Ramamohan Paturi.
\newblock On the complexity of $k$-{SAT}.
\newblock {\em J. Comput. Syst. Sci.}, 62(2):367--375, 2001.
\newblock \href {http://dx.doi.org/10.1006/jcss.2000.1727}
  {\path{doi:10.1006/jcss.2000.1727}}.

\bibitem{ImpagliazzoPZ01}
Russell Impagliazzo, Ramamohan Paturi, and Francis Zane.
\newblock Which problems have strongly exponential complexity?
\newblock {\em J. Comput. Syst. Sci.}, 63(4):512--530, 2001.
\newblock \href {http://dx.doi.org/10.1006/jcss.2001.1774}
  {\path{doi:10.1006/jcss.2001.1774}}.

\bibitem{JansenW16}
Bart M.~P. Jansen and Jules J. H.~M. Wulms.
\newblock Lower bounds for protrusion replacement by counting equivalence
  classes.
\newblock In Jiong Guo and Danny Hermelin, editors, {\em Proc. 11th IPEC},
  volume~63 of {\em LIPIcs}, pages 17:1--17:12. Schloss Dagstuhl -
  Leibniz-Zentrum fuer Informatik, 2016.
\newblock \href {http://dx.doi.org/10.4230/LIPIcs.IPEC.2016.17}
  {\path{doi:10.4230/LIPIcs.IPEC.2016.17}}.

\bibitem{JansenN17}
Bart~M.P. Jansen and Jesper Nederlof.
\newblock Computing the chromatic number using graph decompositions via matrix
  rank.
\newblock In {\em Proc. 26th ESA}, 2018.
\newblock In press.

\bibitem{KorachS93}
Ephraim Korach and Nir Solel.
\newblock Tree-width, path-width, and cutwidth.
\newblock {\em Discrete Applied Mathematics}, 43(1):97--101, 1993.
\newblock \href {http://dx.doi.org/10.1016/0166-218X(93)90171-J}
  {\path{doi:10.1016/0166-218X(93)90171-J}}.

\bibitem{LokshtanovMS11}
Daniel Lokshtanov, D\'aniel Marx, and Saket Saurabh.
\newblock Known algorithms on graphs of bounded treewidth are probably optimal.
\newblock In {\em Proc. 22nd SODA}, pages 777--789, 2011.
\newblock \href {http://dx.doi.org/10.1137/1.9781611973082.61}
  {\path{doi:10.1137/1.9781611973082.61}}.

\bibitem{Marx13}
D{\'{a}}niel Marx.
\newblock The square root phenomenon in planar graphs.
\newblock In Michael~R. Fellows, Xuehou Tan, and Binhai Zhu, editors, {\em
  Proc. 3rd FAW-AAIM}, volume 7924 of {\em Lecture Notes in Computer Science},
  page~1. Springer, 2013.
\newblock \href {http://dx.doi.org/10.1007/978-3-642-38756-2_1}
  {\path{doi:10.1007/978-3-642-38756-2_1}}.

\bibitem{Pilipczuk11}
Michal Pilipczuk.
\newblock Problems parameterized by treewidth tractable in single exponential
  time: A logical approach.
\newblock In {\em Proc. 36th MFCS}, pages 520--531, 2011.
\newblock \href {http://dx.doi.org/10.1007/978-3-642-22993-0_47}
  {\path{doi:10.1007/978-3-642-22993-0_47}}.

\bibitem{ThilikosSB05}
Dimitrios~M. Thilikos, Maria~J. Serna, and Hans~L. Bodlaender.
\newblock Cutwidth {I:} {A} linear time fixed parameter algorithm.
\newblock {\em J. Algorithms}, 56(1):1--24, 2005.
\newblock \href {http://dx.doi.org/10.1016/j.jalgor.2004.12.001}
  {\path{doi:10.1016/j.jalgor.2004.12.001}}.

\bibitem{ThilikosSB05a}
Dimitrios~M. Thilikos, Maria~J. Serna, and Hans~L. Bodlaender.
\newblock Cutwidth {II:} {A}lgorithms for partial w-trees of bounded degree.
\newblock {\em J. Algorithms}, 56(1):25--49, 2005.
\newblock \href {http://dx.doi.org/10.1016/j.jalgor.2004.12.003}
  {\path{doi:10.1016/j.jalgor.2004.12.003}}.

\bibitem{RooijBR09}
Johan M.~M. van Rooij, Hans~L. Bodlaender, and Peter Rossmanith.
\newblock Dynamic programming on tree decompositions using generalised fast
  subset convolution.
\newblock In {\em Proc. 17th ESA}, pages 566--577, 2009.
\newblock \href {http://dx.doi.org/10.1007/978-3-642-04128-0_51}
  {\path{doi:10.1007/978-3-642-04128-0_51}}.

\end{thebibliography}

\clearpage
\appendix

\section{Lower bound for Independent Set on graphs of bounded cutwidth}
\begin{theorem} \label{thm:ind:set:cutwidth:lb}
Assuming SETH, there is no~$\varepsilon > 0$ such that \textsc{Independent Set} on a (nonplanar) graph~$G$ given along with a linear layout of cutwidth~$k$ can be solved in time~$\Oh^*((2-\varepsilon)^k)$.
\end{theorem}
\begin{proof}
This follows from the lower bound of Lokshtanov et al.~\cite[Thm. 3.1]{LokshtanovMS11} in terms of pathwidth. It suffices to extend the analysis and provide an analogue of their Lemma 3.3 to bound the cutwidth of the graph~$G$ that is constructed from an $n$-variable CNF formula by~$n + \Oh(1)$. Graph~$G$ consists of~$n+1$ copies of a graph~$G_1$; the copies are connected in a path-like fashion. We first recall the structure of~$G_1$ and bound its cutwidth.

Graph~$G_1$ consists of~$n$ paths~$P_1, \ldots, P_n$ of~$2m$ vertices each, called~$p_i^1, \ldots, p_i^{2m}$ for~$i \in [n]$, together with~$m$ clause gadgets~$\hat{C}_j$ for~$j \in [m]$. Each clause gadget is easily seen to be a graph of cutwidth~$\Oh(1)$. Between a clause gadget~$\hat{C}_j$ and a path~$P_i$ there is at most one edge, which connects to either~$p_i^{2j-1}$ or~$p_i^{2j}$. Each vertex of a clause gadget is adjacent to at most one vertex of one path. 

Using this knowledge we describe a linear layout~$\pi$ of~$G_1$ of cutwidth~$n + \Oh(1)$. It consists of~$m$ consecutive blocks~$B_1, \ldots, B_m$ of vertices. A block~$B_j$ contains the vertices of~$\hat{C}_j \cup \{ p_i^{2j-1}, p_i^{2j} \mid i \in [n]\}$, and is ordered according to the following process. Start from an optimal ordering for~$\hat{C}_j$, of cutwidth~$\Oh(1)$. For every vertex~$v$ of~$\hat{C}_j$ that is adjacent to a vertex on a path, say to~$P_i$, insert vertices~$p_i^{2j-1}, p_i^{2j}$ just after~$v$ in the ordering. For the paths that are not adjacent to~$\hat{C}_j$, put their two vertices~$p_i^{2j-1}, p_i^{2j}$ next to each other at the end of the block; the order among these pairs is not important. To see that the resulting ordering~$\pi$ has cutwidth~$n + \Oh(1)$, note that from each path~$P_i$, the vertices on~$P_i$ appear along~$\pi$ in their natural order. Hence for any vertex~$v \in V(G_1)$, the gap after vertex~$v$ is crossed by at most one edge from each path~$P_i$. A gap is crossed by at most one clause gadget, whose internal edges contribute~$\Oh(1)$ to the size of the cut. Finally, there is at most one edge from a clause gadget to a path that crosses the cut after~$v$: a vertex from~$\hat{C}_j$ that has a neighbor on a path~$P_i$ is immediately followed by two vertices from~$P_i$ that include its neighbor, removing that edge from later cuts. This proves that~$\cutw_{\pi}(G_1) \leq n + \Oh(1)$.

To see that~$\cutw_{\pi}(G) \leq n + \Oh(1)$, we note that~$G$ is obtained from~$n+1$ copies~$G_1, G_2, \ldots, G_{n+1}$ of~$G_1$ by connecting the last vertex on the $j$-th path in~$G_i$, to the first vertex of the $j$-th path in~$G_{i+1}$, for all~$i \in [n]$ and~$j \in [n]$. Hence the number of edges connecting any~$G_i$ to vertices in later copies is at most~$n$. From this it follows that by simply constructing the order~$\pi$ for each graph individually, and concatenating these, we obtain a linear ordering of~$G$ of cutwidth~$n + \Oh(1)$.
\end{proof}

\section{Lower bound for Dominating Set on graphs of bounded cutwidth}

\subsection{Proof of Theorem \ref{thm:domset:nonplanar:lb}}
{
\renewcommand{\thetheorem}{\ref{thm:domset:nonplanar:lb}}
\begin{theorem}
\nonplanardomsetthm
\end{theorem}
\addtocounter{theorem}{-1}
}
\begin{proof}
	The proof follows the argumentation of Lokshtanov et al.~\cite[Thm 4.1]{LokshtanovMS11} with small modifications along the way. To avoid having to repeat the entire proof, at several steps we only describe how to modify the existing construction.
	
	Suppose that there is an~$\varepsilon > 0$ such that \textsc{Dominating Set} on graphs given with a linear layout of cutwidth~$k$ can be solved in time~$\Oh^*((3-\varepsilon)^k)$. We will prove that SETH is false, by showing that it implies the existence of~$\delta > 0$ such that $n$-variable \qSAT can be solved in time~$\Oh^*((2-\delta)^n)$ for each fixed~$q$.\footnote{This is a somewhat weaker consequence than used by Lokshtanov et al., who obtain the consequence that \textsc{CNF-SAT} for clauses of \emph{arbitrary} size can be solved by a uniform algorithm in time~$\Oh^*((2-\delta)^n)$ for some~$\delta > 0$. By making more major modifications to the construction we could arrive at the same consequence, but for ease of presentation we will simply show that SETH fails.} We choose an integer~$p$ depending on~$\varepsilon$ in a manner that is described at the end of the proof. Consider now an $n$-variable input formula~$\phi$ of \qSAT for some constant~$q$. Let~$C_1, \ldots, C_m$ be the clauses of~$\phi$. We split the variables of~$\phi$ into groups~$F_1, \ldots, F_t$, each of size at most~$\beta := \lfloor \log 3^p \rfloor = \lfloor p \log 3 \rfloor$ so that~$t = \lceil n/\beta \rceil$. (Recall that all logarithms in this paper have base two.)
	
	We now follow the construction of Lokshtanov et al.~to produce a graph~$G$ that has a dominating set of size~$\ell := (p+1)tm(2pt+1)+1$ if and only if~$\phi$ is satisfiable (\cite[Lemmas 4.1, 4.2]{LokshtanovMS11}). We then modify the graph~$G$ slightly to obtain~$G'$ which has a dominating set of size~$\ell+1$ if and only if~$G$ has a dominating set of size~$\ell$. To describe the modification, we summarize the essential features of the graph~$G$ built in the original construction.
	
	The graph~$G$ consists of \emph{group gadgets}~$\widehat{B}^j_i$ for~$i \in [t]$ and~$j \in [m(2pt+1)]$, of \emph{clause vertices}~$\hat{c}^\ell_i$ for~$j \in [m]$ and~$0 \leq i < 2pt+1$, and of two special vertices~$h$ and~$h'$ (see~\cite[Figure 3]{LokshtanovMS11}). Each group gadget contains~$\Oh(3^p)$ vertices. It has~$p$ \emph{entry vertices} and~$p$ \emph{exit vertices}; these~$2p$ vertices are all distinct. For each~$i \in [t]$ and~$j \in [m(2pt+1) - 1]$ there is a matching between the exit vertices of~$\widehat{B}^{j}_i$ and the entry vertices of~$\widehat{B}^{j+1}_i$. There are no other edges between group gadgets. Each clause vertex~$\hat{c}^\ell_j$ is adjacent to at most~$q$ different group gadgets, corresponding to the groups that contain the literals in clause~$C_j$. The group gadgets to which $\hat{c}^\ell_j$ is adjacent belong to~$\{\widehat{B}^{m\ell + j}_i \mid i \in [t]\}$. Finally, vertex~$h'$ has degree one and is adjacent to~$h$. The other neighbors of~$h$ are the entry vertices of~$\{ \widehat{B}^1_i \mid i \in [t]\}$, and the exit vertices of~$\{\widehat{B}^{2pt+1}_i \mid i \in [t]\}$.
	
	With this summary of~$G$, our modification to obtain~$G'$ can be easily described. We remove vertices~$h$ and~$h'$ and replace them by~$h_1, h_2, h'_1, h'_2$. Vertices~$h'_1$ and~$h'_2$ have degree one and are adjacent to~$h_1$ and~$h_2$, respectively. Vertex~$h_1$ is further adjacent to the entry vertices of the group gadgets~$\{ \widehat{B}^1_i \mid i \in [t]\}$, and vertex~$h_2$ is further adjacent to the exit vertices of the group gadgets~$\{\widehat{B}^{2pt+1}_i \mid i \in [t]\}$. Essentially, we have split the vertex~$h$ into two vertices to reduce the cutwidth of the graph. Note that~$\optDS(G') = \optDS(G) + 1$. The presence of the degree-one vertices ensures that there is always a minimum dominating set of~$G$ that contains~$h$, and of~$G'$ that contains both~$h_1$ and~$h_2$. Such dominating sets may be transformed into one another by exchanging~$h$ with~$h_1$ and~$h_2$, since~$h$ dominates the same as~$h_1$ and~$h_2$ combined. Hence we find that~$G'$ has a dominating set of size~$\ell+1$ if and only if~$\phi$ is satisfiable.
	
	We proceed to bound the cutwidth of~$G'$. For~$j \in [2pt+1]$, let the \emph{$j$-th column} of~$G'$ consist of the group gadgets~$\{\widehat{B}^{j}_i \mid i \in [t]\}$, together with the unique clause vertex that is adjacent to those group gadgets. The linear layout~$\pi'$ of~$G'$ starts with vertices~$h'_1$ and~$h_1$. Then, for each~$j \in [2pt+1]$, it first has the clause vertex of the $j$-th column, followed by the contents of the group gadgets in that column, one gadget at a time. It ends with~$h_2$ and finally~$h'_2$.
	
	\begin{claim} \label{claim:domset:cutwidth}
		$\cutw_{\pi'}(G') \leq tp + \Oh(q \cdot 3^p + (3^p)^2)$.
	\end{claim}
	\begin{claimproof}
		Consider an arbitrary vertex~$v^* \in V(G')$ and the cut consisting of the edges crossing the gap after~$v^*$ in layout~$\pi'$. The cut after~$h'_1$ or~$h_2$ has size one. The cut after~$h_1$ consists of the~$tp$ edges to the entry vertices of the group gadgets in the first column, and the cut after~$h'_2$ is empty. It remains to consider the case that~$v^*$ belongs to some column~$j$. 
		
		If~$v^*$ is the clause gadget of column~$j$, then the cut after~$v^*$ consists of the edges from~$v^*$ to its neighbors in the group gadgets in that column, together with the~$tp$ edges on the~$t$ matchings of size~$p$ that connect the group gadgets in column~$j-1$ to the group gadgets in column~$j$. Since each group gadget has~$\Oh(3^p)$ vertices and a clause vertex is adjacent to at most~$q$ different group gadgets, it follows that the size of the cut after~$v^*$ is~$tp + \Oh(q \cdot 3^p)$.
		
		If~$v^*$ is not the clause gadget of column~$j$, then it belongs to some group gadget~$\widehat{B}^j_i$ of column~$j$. For all other group gadgets, its vertices occur on the same side of~$v^*$ in the ordering. Hence the cut after~$v^*$ contains edges that are internal to at most one group gadget. Since a group gadget has~$\Oh(3^p)$ vertices, it has~$\Oh((3^p)^2)$ edges that can be contributed to the cut. For all group gadgets in column~$j$ that appeared before~$\widehat{B}^j_i$ in the ordering, the~$p$ matching edges from their exit vertices to the entry vertices of the next column (or to~$h_2$) belong to the cut. Similarly, for the group gadgets in column~$j$ that appear after~$\widehat{B}^j_i$, the~$p$ matching edges from their entry vertices to the exits of the previous column belong to the cut. For~$\widehat{B}^j_i$ itself, there are at most~$2p$ edges connecting to other columns in the cut. The only other edges that can be in the cut are from the group gadgets of column~$j$ to its clause gadget, and as argued above there are~$\Oh(q \cdot 3^p)$ of those. It follows that the size of the cut after~$v^*$ is at most~$tp + \Oh(q \cdot 3^p + (3^p)^2)$.
	\end{claimproof}
	
	Using this construction of~$G'$ and linear layout~$\pi'$ we complete the proof. Suppose that \textsc{Dominating Set} on graphs with a given linear layout of cutwidth~$k$ can be solved in~$\Oh^*((3-\varepsilon)^k) = \Oh^*(3^{\lambda k})$ time, for~$\lambda < 1$. We choose~$p$ large enough that~$\lambda \cdot \frac{p}{\lfloor p \log 3 \rfloor} \leq \frac{\delta}{\log 3}$ for some~$\delta < 1$. Choose a function~$f(p,q)$ such that the cutwidth in Claim~\ref{claim:domset:cutwidth} is bounded by~$tp + f(p,q)$. Then an instance~$\phi$ of \qSAT for fixed~$q$ can be solved by transforming it into an instance of \textsc{Dominating Set} in time polynomial in~$\phi + 3^p$ and applying the assumed algorithm in time 
\allowdisplaybreaks
	\begin{align*}
\Oh^*(3^{\lambda (tp + f(p,q))}) & \leq \Oh^*(3^{\lambda p \lceil n/\beta \rceil}) & \text{$\lambda \cdot f(p,q) \in \Oh(1),~t = \lceil n/\beta \rceil$} \\ 
& = \Oh^*(3^{\lambda p \lceil \frac{n}{\lfloor p \log 3 \rfloor} \rceil }) & \text{$\beta = \lfloor p \log 3 \rfloor$} \\
& \leq \Oh^*(3^{\lambda p (\frac{n}{\lfloor p \log 3 \rfloor} + 1)}) \leq \Oh^*(3^{\lambda p \frac{n}{\lfloor p \log 3 \rfloor}}) & \text{$\lambda p \in \Oh(1)$} \\
& \leq \Oh^*(3^{\frac{\delta n}{\log 3}) }) \leq \Oh^*(2^{\delta n}). & \text{choice of~$\delta$}
	\end{align*}
	We used the fact that since~$p$ and~$q$ are constants, their contributions can be absorbed into the~$\Oh^*$ notation. Since this shows that \qSAT for any constant~$q$ can be solved in time~$\Oh^*(2^{\delta n})) = \Oh^*((2-\delta')^n)$ for some~$\delta' < 1$, this contradicts SETH and concludes the proof of Theorem~\ref{thm:domset:nonplanar:lb}.
\end{proof}

\end{document}